\documentclass[12pt]{amsart}
\usepackage{amsmath,amscd,amsfonts,amssymb, color,bbm, bm,
braket}
\usepackage{latexsym, graphicx, pstricks,rotating, enumerate}
\usepackage[pdftex,bookmarks,colorlinks]{hyperref}
\definecolor{darkblue}{rgb}{0.0,0.0,0.3}
\hypersetup{colorlinks=true,citecolor=darkblue,
unicode=true,pdftitle=extend.pdf,pdfauthor=Ritabrata
Sengupta,bookmarks=false, urlcolor=darkblue} 
\numberwithin{equation}{section}
\newtheorem{prop}{Proposition}[section]
\newtheorem{defin}{Definition}[section]
\newtheorem{rem}{Remark}[section]

\newtheorem{theorem}{Theorem}[section]
\newtheorem{lemma}{Lemma}[section]

\newtheorem{innercustomthm}{Theorem}
\newenvironment{customthm}[1]
  {\renewcommand\theinnercustomthm{#1}\innercustomthm}
  {\endinnercustomthm}
\newcommand{\tr}{\mathrm{Tr\,}}
\newcommand{\dd}{\mathrm{d}}
\newcommand{\re}{\mathrm{Re\,}}

\newcommand{\ketbra}[1]{\ensuremath{\left|#1\right\rangle
\hspace{-3pt}\left\langle #1\right|}}
\begin{document}
\title{On the equivalence of separability and extendability of quantum states}
\author{B. V. Rajarama Bhat}
\address{Theoretical Statistics and Mathematics Unit, Indian
Statistical Institute, Bengalore Centre, 8th Mile, Mysore Road
RVCE Post, Bangalore 560 059}
\email[B. V. Rajarama
Bhat]{\href{mailto:bhat@isibang.ac.in}{bhat@isibang.ac.in}}
\author{K. R. Parthasarathy}
\address{Theoretical Statistics and Mathematics Unit, Indian
Statistical Institute, Delhi Centre, 7 S J S Sansanwal Marg,
New Delhi 110 016, India}
\email[K R Parthasarathy]{\href{mailto:krp@isid.ac.in}{krp@isid.ac.in}}
\author{Ritabrata Sengupta}
\address{Theoretical Statistics and Mathematics Unit, Indian
Statistical Institute, Delhi Centre, 7 S J S Sansanwal Marg,
New Delhi 110 016, India}
\email[Ritabrata Sengupta]{\href{mailto:rb@isid.ac.in}{rb@isid.ac.in}}
\begin{abstract}
 
\par Motivated by the notions of $k$-extendability and complete
extendability of the state of a  finite level quantum system as
described by Doherty et al (Phys. Rev. A, 69:022308),
we introduce parallel definitions in the context of 
Gaussian states and using only properties of their
covariance matrices derive necessary and sufficient
conditions for their complete extendability. It turns out
that the complete extendability property is equivalent to
the separability property of a bipartite Gaussian state.

\par Following the proof of quantum de Finetti
theorem as outlined in Hudson and Moody (Z.
Wahrscheinlichkeitstheorie und Verw. Gebiete,
33(4):343--351), we show that separability is equivalent to
complete extendability for a state in a bipartite Hilbert
space where at least one of which is of dimension greater
than 2. This, in particular, extends the result of Fannes,
Lewis, and Verbeure (Lett. Math. Phys. 15(3): 255--260) to
the case of an infinite dimensional Hilbert space whose C*
algebra of all bounded operators is not separable.


\smallskip

\end{abstract}
\thanks{The authors thank Professor Ajit Iqbal Singh for
useful comments and suggestions. RS acknowledges financial
support from the National Board for Higher Mathematics,
Govt. of India. This work is supported by the \emph{Advances in
Non-Commutative Mathematics} (ANCM) Project funded by ISI.} 
\maketitle
\section{Introduction}
\par One of the most important problems in quantum mechanics
as well as quantum information theory is to determine whether a
given bipartite state is separable or entangled \cite{NC}.
There are several methods in tackling this problem leading
to a long list of important publications. A detailed
discussion on this topic is available in the survey articles
by Horodecki et al \cite{RevModPhys.81.865}, and 
G{\"u}hne and T{\'o}th \cite{guth}. One such condition which
is both necessary and sufficient for separability in finite
dimensional product spaces is complete extendability
\cite{PhysRevA.69.022308}.

\begin{defin}\label{def:1}
Let $k \in \mathbb{N}$. A state $\rho \in
\mathcal{B}(\mathcal{H}_A \otimes \mathcal{H}_B)$ is said to
be $k$-extendable with respect to system $B$ if there is a
state $\tilde{\rho} \in \mathcal{B}(\mathcal{H}_A \otimes
\mathcal{H}_B^{\otimes k})$ which is invariant under any
permutation in $\mathcal{H}_B^{\otimes k}$ and $\rho =
\tr_{\mathcal{H}_B^{\otimes (k-1)}}\tilde{\rho}$, $k \ge 2$. 
\par A state $\rho \in \mathcal{B}(\mathcal{H}_A \otimes
\mathcal{H}_B)$ is said to be completely extendable if it
is $k$-extendable for all $k \in \mathbb{N}$. 
\end{defin}
The following theorem of Doherty, Parrilo, and
Spedalieri \cite{PhysRevA.69.022308} emphasizes the
importance of the notion of complete extendability.

\begin{customthm}{A}\cite{PhysRevA.69.022308} \label{th:a}
A bipartite state $\rho \in \mathcal{B}(\mathcal{H}_A \otimes
\mathcal{H}_B)$ is separable if and only if it is completely
extendable with respect to one of its subsystems.
\end{customthm}

\par It is fairly simple to see that separability implies complete
extendability. The proof of the converse depends
on an application of the quantum de Finetti theorem
\cite{MR0241992, MR0397421}, according to which any
exchangable state is, indeed, separable. The link between separability and
extendability has found applications in quantum information theory
\cite{MR2825510, MR3210848}. Here we study the same in
the context of quantum Gaussian states.

\par The importance of finite mode Gaussian states and their
covariance matrices in general quantum theory as well as
quantum information has been highlighted extensively in the
literature. A comprehensive survey of Gaussian states and
their properties can be found in the
book of Holevo \cite{MR2797301}. For their applications to 
quantum information theory the reader is referred to the
survey article by Weedbrook et al \cite{RevModPhys.84.621}
as well as Holevo's book \cite{ MR2986302}.  For our
reference we use \cite{arvindsurvey, MR2662722,
MR3070484} for Gaussian states and for notations in the
following sections we use \cite{krprb} and \cite{krprb2}.

\par If $\rho$ is a state of a quantum system and
$X_i,\,i=1,2$ are two real-valued observables, or
equivalently, self-adjoint operators with finite second
moments in the state $\rho$ then the covariance between $X_1$
and $X_2$ in the state $\rho$ is the scalar quantity
\[\tr\left(\frac{1}{2}(X_1 X_2 +X_2 X_1)\rho\right) - \left(\tr X_1
\rho\right) \cdot \left(\tr X_2\rho\right),\]
 which is denoted by
$\mathrm{Cov}_\rho (X_1,X_2)$. Suppose $q_1,p_1;\,
q_2,p_2;\, \cdots ;\, q_n,p_n$ are the position - momentum
pairs of observables of a quantum system with $n$ degrees of
freedom obeying the canonical  commutation relations. Then
we express 
\[(X_1,X_2,\cdots,X_{2n}) =
(q_1,-p_1,q_2,-p_2,\cdots,q_n,-p_n).\]
If $\rho$ is a state in which all the $X_j$'s have finite
second moments we write 
\begin{equation}\label{eq:1.1}
S_\rho = [[ \mathrm{Cov}_\rho(X_i, X_j)]], \quad i,j \in \{
1,2,\cdots, 2n\}.
\end{equation} 
We call $S_\rho$ the covariance matrix of the position
momentum observables. If we write 
\begin{equation}\label{eq:1.2}
J_{2n}= \begin{bmatrix}\begin{array}{rr}
0 & 1 \\ -1 & 0
\end{array} &&& \\
& \begin{array}{rr}
0 & 1 \\ -1 & 0
\end{array} &&\\ 
 &  & \ddots & \\
&&& \begin{array}{rr}
0 & 1 \\ -1 & 0
\end{array}
\end{bmatrix}
\end{equation}
or equivalently $\bigoplus_1^n \begin{bmatrix} 0 & 1 \\ -1 &
0 \end{bmatrix} $ for the $2n \times 2n$ block diagonal
matrix, the complete Heisenberg uncertainty relations for
all the position and momentum observables assume the form
of the following matrix inequality 
\begin{equation}\label{eq:1.3}
S_\rho +  \frac{\imath}{2} J_{2n} \geq 0.
\end{equation} 
Conversely, if $S$ is any real $2n \times 2n$ symmetric
matrix obeying the inequality $S +  \frac{\imath}{2}
J_{2n} \geq 0$, then there exists a state $\rho$ such that
$S$ is the covariance matrix $S_\rho$ of the observables
$q_1,-p_1;\,q_2,-p_2;\, \cdots ;\, q_n,-p_n$. In such a case $\rho$ can be
chosen to be a Gaussian state with mean zero. 
Recall \cite{MR2662722}, a state $\rho$ in $\Gamma(\mathcal{H})$ with $\mathcal{H}=
\mathbb{C}^n$ is an $n$-mode Gaussian state if its
Fourier transform $\hat{\rho}$ is given by 
\begin{equation}\label{eq:1.4}
\hat{\rho}(\bm{x} + \imath \bm{y}) = \exp
\left[-\imath \sqrt{2} (\bm{l}^T\bm{x} -
\bm{m}^T\bm{y}) - \begin{pmatrix} \bm{x}\\
\bm{y} \end{pmatrix}^T S  \begin{pmatrix} \bm{x}\\
\bm{y} \end{pmatrix} \right].
\end{equation}
for all $\bm{x},~\bm{y}\in \mathbb{R}^n$ where $\bm{l},~
\bm{m}$ are the momentum-position mean vectors and $S$ their
covariance matrix.

\par We organise the paper as follows. Motivated by the
concept of extendability for finite dimensional
systems, we define Gaussian extendability for Gaussian
states in \S \ref{Sec:2}. We study this extendability
problem entirely in terms of covariance matrices.  
In \S \ref{sec:3} we look at the
same problem in an abstract way.  We show that for
a bipartite state of arbitrary dimension, complete
extendability is equivalent to separability.
This proof directly follows from work of Hudson and Moody
\cite{MR0397421}. The problem of finding necessary and
sufficient conditions for $k$-extendability of states in
both Gaussian and non-Gaussian cases remains open.

\section{Gaussian extendability} \label{Sec:2}

\begin{defin}[Gaussian extendability]\label{def:2}
Let $k \in \mathbb{N}$. A Gaussian state $\rho_g$ in
$\Gamma(\mathbb{C}^m) \otimes \Gamma(\mathbb{C}^n)$ is said to
be Gaussian $k$-extendable with respect to the second system  if there is a
Gaussian state $\tilde{\rho_g}$ in  $\Gamma(\mathbb{C}^m) \otimes
\Gamma(\mathbb{C}^n)^{\otimes k}$  which is invariant under any
permutation in $\Gamma(\mathbb{C}^n)^{\otimes k}$ and $\rho_g =
\tr_{\Gamma(\mathbb{C}^n)^{\otimes (k-1)}}\tilde{\rho_g}$,
$k\ge 2$. 
\par A Gaussian state $\rho_g$ in $\Gamma(\mathbb{C}^m)
\otimes \Gamma(\mathbb{C}^n)$ is said to be Gaussian
completely extendable if it is Gaussian $k$-extendable
for every $k \in \mathbb{N}$.
\end{defin}
\begin{rem}
In this section we confine our attention to Gaussian states
only and so we use the terms $k$-extendability
and complete extendability to mean Gaussian $k$-extendability
and Gaussian complete extendability respectively, unless
stated otherwise. In the next section \S \ref{sec:3}, we use
extendability and complete extendability in its usual sense. 
\end{rem}

\par We shall use the following result.
\begin{customthm}{B}\label{cu:B}
Let 
\[X=\left[\begin{array}{c|c}
A & B\\ \hline 
B^\dag & C\end{array}\right]
\]
be a Hermitian block matrix with real or complex entries,
$A$ and $C$ being strictly positive
matrices of order $m \times m$ and $n\times n$ respectively. 
Then $X\ge 0$ if and only if 
\[A\ge B C^{-1} B^\dag.\]
\end{customthm}
\begin{proof}
For a proof, see Theorem 1.3.3 in the book of Bhatia
\cite{bhatia1}.
\end{proof}
\par Entanglement property of a Gaussian state depends only on
its covariance matrix. Hence without loss of generality, we
can confine our attention to the Gaussian states with
mean zero. Thus an $(m+n)$-mode mean zero Gaussian state in
$\Gamma(\mathbb{C}^m) \otimes \Gamma(\mathbb{C}^n)$ is
uniquely determined by a $2(m+n) \times 2(m+n)$ covariance
matrix 
\[S = \begin{bmatrix} A & B \\ B^T & C \end{bmatrix}.\]
Here $A$ and $C$ are covariance matrices of the $m$ and $n$-mode
marginal states respectively. 

\par If $\rho(\bm{0},\bm{0};S)$, written in short as
$\rho(S)$ in $\Gamma(\mathbb{C}^m) \otimes
\Gamma(\mathbb{C}^n)$ is $k$-extendable with respect to the
second system, then there exists a real matrix $\theta_k$ of
order $2n \times 2n$ such that the extended matrix
\begin{equation}\label{eq:1.5}
S_k = \left[ \begin{array}{c|cccc}
A & B & B & \cdots & B \\\hline
B^T & C & \theta_k & \cdots & \theta_k \\
B^T & \theta_k^T & C & \cdots & \theta_k\\
\vdots & \vdots & \vdots & \ddots & \vdots\\
B^T & \theta_k^T & \theta_k^T & \cdots & C
\end{array} \right]
\end{equation} 
is the covariance matrix of a Gaussian state in
$\Gamma(\mathbb{C}^m) \otimes \Gamma(\mathbb{C}^n)^{\otimes
k}$. Then it satisfies inequality (\ref{eq:1.3}) in the form
\begin{equation}\label{eq:1.6}
S_k + \frac{\imath}{2} J_{2(m+kn)} \ge 0.
\end{equation}
Now we observe that $\theta_k$ can be chosen independent of
$k$. To prove this we need the following theorem which will
also be used later in this paper. 
\begin{customthm}{C}\label{cu:C}
Let $A$ and $B$ be positive. Then the matrix 
$\left[\begin{array}{c|c}
A & B\\ \hline 
B^\dag & C\end{array}\right]$
is positive if and only if $B=A^{\frac{1}{2}} K
C^{\frac{1}{2}}$ for some contraction $K$.
\end{customthm}
\begin{proof}
For a proof, see Proposition 1.3.2 in the book of Bhatia
\cite{bhatia1}.
\end{proof}
The reason follows form the fact that the marginal
covariance matrix $\begin{bmatrix} C & \theta_k \\
\theta_k^T & C\end{bmatrix} \ge 0$. Using the above Theorem
\ref{cu:C}, this is equivalent to the existence of a
contraction $K$ with $\|K\|\le1$ such that $\theta_k
= C^{\frac{1}{2}} K C^{\frac{1}{2}}$. Hence $\|\theta_k\|
\le \|C\|$. Since for a given state, $C$ is fixed, we have
$\theta_k$ bounded. Hence, in the set of all $\theta_k$'s
there is a convergent subsequence which converges to a
$\theta$ and can replace $\theta_k$ in (\ref{eq:1.5}) by
$\theta$. The extension matrix for each $k=1,2,3,\cdots$
will look like 
 \begin{equation}\label{eq:1.5a}
S_k = \left[ \begin{array}{c|cccc}
A & B & B & \cdots & B \\\hline
B^T & C & \theta_k & \cdots & \theta_k \\
B^T & \theta_k^T & C & \cdots & \theta_k\\
\vdots & \vdots & \vdots & \ddots & \vdots\\
B^T & \theta_k^T & \theta_k^T & \cdots & C
\end{array} \right]
\end{equation}  
Hence by Definition \ref{def:1}, $\rho(S)$ is completely
extendable if the inequality (\ref{eq:1.6}) holds for every
$k=1,\,2,\,\cdots.$
\par Let us denote the marginal covariance matrix
corresponding to  $\Gamma(\mathbb{C}^n)^{\otimes k}$ by 
\[\Sigma_k(C,\theta) = \begin{bmatrix}
C & \theta & \cdots & \theta \\
\theta^T & C & \cdots & \theta\\
 \vdots & \vdots & \ddots & \vdots\\
 \theta^T & \theta^T & \cdots & C
\end{bmatrix}.\]
If $\rho$ is completely extendable, $S_k$ is a
covariance matrix for each $k$, and hence
$\Sigma_k(C,\theta)$ is a covariance matrix for each $k$ as
well. Using Theorem 1 of \cite{krprb2} (see also
\cite{MR2492582}), such a pair
$(C,\theta)$ defines a covariance matrix $\Sigma_k(C,\theta)$
for each $k=1,2,3,\cdots$ if and only if 
\begin{enumerate}[(i)]
\item $\theta$ is a real symmetric positive semidefinite
matrix, and 
\item $C - \theta +\frac{\imath}{2}J_{2n} \ge 0$.
\end{enumerate}
In particular, $S_k$ is of the form 
\begin{equation}\label{eq:1.7}
 S_k = \left[ \begin{array}{c|cccc}
A & B & B & \cdots & B \\\hline
B^T & C & \theta & \cdots & \theta \\
B^T & \theta & C & \cdots & \theta\\
\vdots & \vdots & \vdots & \ddots & \vdots\\
B^T & \theta & \theta & \cdots & C
\end{array} \right],
\end{equation} 
where $\theta$ is a real positive semidefinite matrix. 

\par Our first theorem gives a necessary and sufficient
condition for complete extendability of Gaussian states. 

\begin{lemma}\label{le1}
Let $\rho$ be a bipartite Gaussian state in  $\Gamma(\mathbb{C}^m)
\otimes \Gamma(\mathbb{C}^n)$ with no pure
marginal state in $\Gamma(\mathbb{C}^m)$ as well as
$\Gamma(\mathbb{C}^n)$. Let $S =
\begin{bmatrix} A & B \\ B^T & C \end{bmatrix}$ be the
covariance matrix of $\rho$, where $A$
and $C$ are marginal covariance matrices of the first and second
system respectively. Then $\rho$ is completely extendable with respect
to the second system if and only if there exists a real positive
matrix $\theta$ such that 
\begin{equation}\label{eq:x} 
C+\frac{\imath}{2} J_{2n} \ge \theta \ge B^T \left( A + \frac{\imath}{2} J_{2m}
\right)^{-1} B.
\end{equation} 
\end{lemma}

\begin{proof}
\par Without loss of generality, we may assume that $A$ and
$C$ are written in their Williamson normal forms. Since no
pure state is a marginal of $\rho$, $\frac{1}{2}I_2$ is not
a sub-matrix of ether $A$ or $C$. This implies $\left(A +
\frac{\imath}{2} J_{2m}\right)$ and $\left(C +
\frac{\imath}{2} J_{2n}\right)$ are invertible, and hence we
can apply Theorem \ref{cu:B}, when $A$ and $C$ are replaced
respectively by $\left(A +
\frac{\imath}{2} J_{2m}\right)$ and $\left(C +
\frac{\imath}{2} J_{2n}\right)$.  
Thus, 
\[ C+\frac{\imath}{2} J_{2n} \ge B^T \left( A +
\frac{\imath}{2} J_{2m} \right)^{-1} B.\]

\par The necessity of the left part of inequality
\ref{eq:x} is already contained in the
discussion above (\ref{eq:1.7}). Hence, all we need to 
prove is the right part the same inequality starting
from (\ref{eq:1.6}).
 
\par Setting $\ket{\psi_k}
=\frac{1}{\sqrt{k}}[1,1,\cdots,1]^T\in \mathbb{C}^k $ and 
$\sqrt{k}\mathcal{B}_k=B \otimes \bra{\psi_k}$, the left
hand side of (\ref{eq:1.6}) can be expressed as
\[\begin{bmatrix}
A + \frac{\imath}{2} J_{2m} & \sqrt{k} \mathcal{B}_k \\
&\\
\sqrt{k} \mathcal{B}_k^T & \Sigma_k + \frac{\imath}{2}
J_{2nk}
\end{bmatrix} .\]
By Theorem \ref{cu:B} this matrix is positive if and only if 
\[\Sigma_k + \frac{\imath}{2} J_{2nk} \ge k \mathcal{B}_k^T
\left(A + \frac{\imath}{2} J_{2m}\right)^{-1}
\mathcal{B}_k.\]
By elementary algebra, this is equivalent to
\begin{multline*}
\left( C - \theta + \frac{\imath}{2} J_{2n}
\right)  \otimes (I_k - \ketbra{\psi_k}) + \left ( C +
\overline{k-1} \theta + \frac{\imath}{2} J_{2n} \right)
\otimes \ketbra{\psi_k}\\
 \ge k B^T \left(A
+\frac{\imath}{2}J_{2m} \right)^{-1} B \otimes
\ketbra{\psi_k} .
\end{multline*}
Since $\ketbra{\psi_k}$ and $I_k - \ketbra{\psi_k}$ are
mutually orthogonal projections, it follows that the
inequality above is equivalent to 
\begin{eqnarray*}
&& \left ( C + \overline{k-1} \theta + \frac{\imath}{2} J_{2n} \right)
 \ge k B^T \left(A
+\frac{\imath}{2}J_{2m} \right)^{-1} B,
\end{eqnarray*}
which can be rewritten as 
\begin{equation}\label{eq:a1}
\frac{1}{k}\left ( C - \theta + \frac{\imath}{2} J_{2n}
\right) +\theta
 \ge  B^T \left(A +\frac{\imath}{2}J_{2m} \right)^{-1} B,
\quad \text{ for every } k\in \mathbb{N}.
\end{equation}
Since $\left ( C - \theta +
\frac{\imath}{2} J_{2n} \right)$ is positive and the left hand side
decreases monotonically to $\theta$ as $k \rightarrow
\infty$, it follows that (\ref{eq:a1}) is equivalent to 
\begin{equation*}
\theta \ge B^T \left(A +\frac{\imath}{2}J_{2m} \right)^{-1}
B.
\end{equation*}
\end{proof}  
\par We now consider the case when the Gaussian state
$\rho_g$ admits a pure marginal state. 
\begin{prop}\label{pro1}
If $X=\begin{bmatrix} A & B \\ B^T & \frac{1}{2}I_{2s}
\end{bmatrix}$ is a Gaussian covariance matrix, then $B=0$.
\end{prop}
\begin{proof}
Let $C=\frac{1}{2}(I_{2s} + \imath J_{2s})$. Then $C$ is a
projection with $C^{\frac{1}{2}} = C$. It follows that $C(I_{2s} -
\imath J_{2s})=0$. Since $X$ is a Gaussian matrix, 
\[X +\frac{\imath}{2} J_{2(n+s)} =  \begin{bmatrix} A+ \frac{\imath}{2} J_{2n}& B \\
B^T & \frac{1}{2}(I_{2s} + \imath J_{2s})\end{bmatrix}\ge 0,\] 
there is a contraction
$D$ such that $B = \left( A +\frac{\imath}{2} J_{2n}
\right)^{\frac{1}{2}} D \frac{1}{2}(I_{2s} + \imath
J_{2s})^{\frac{1}{2}} = \left( A +\frac{\imath}{2} J_{2n}
\right)^{\frac{1}{2}} D C$, where $\|D\| \le
1$. Then $B (I_{2s} - \imath J_{2s}) =0$, and hence $B =
\imath \re ( B - B J_{2s})=0$.
\end{proof}

\begin{theorem}\label{th:1}
Let $\rho$ be a bipartite Gaussian state in $\Gamma(\mathbb{C}^m)
\otimes \Gamma(\mathbb{C}^n)$ with covariance matrix $S =
\begin{bmatrix} A & B \\ B^T & C \end{bmatrix}$, where $A$
and $C$ are marginal covariance matrices of the first and second
system respectively. Then $\rho$ is completely extendable with respect
to the second system if and only if there exists a real positive
matrix $\theta$ such that 
\begin{equation}\label{eq:xx} 
C+\frac{\imath}{2} J_{2n} \ge \theta \ge B^T \left( A + \frac{\imath}{2} J_{2m}
\right)^{-} B,
\end{equation} 
where $\left( A + \frac{\imath}{2} J_{2m}\right)^{-}$ is the
Moore-Penrose inverse of $A + \frac{\imath}{2} J_{2m}$.
\end{theorem}

\begin{proof}
Since the case where both $A + \frac{\imath}{2} J_{2m}$ and $C +
\frac{\imath}{2} J_{2n}$ are
invertible has already been dealt with in Lemma \ref{le1}, we only
need to prove in the case when $\rho$ admits pure marginal
states.
\par Without loss of generality let us assume that $A$ and $C$
are written in their Williamson normal forms. Let $A =
\left( \oplus_1^k \kappa_j I_2\right)\bigoplus
\left(\oplus_{k+1}^m \frac{1}{2} I_2\right) = A' \bigoplus
\frac{1}{2} I_{2(m-k)}$ and $C =
\left( \oplus_1^s \mu_l I_2\right)\bigoplus
\left(\oplus_{s+1}^n \frac{1}{2} I_2\right) = C' \bigoplus
\frac{1}{2} I_{2(n-s)}$, where
$\kappa_j,\,\mu_l>\frac{1}{2}$ for every $j,\,l$. By Proposition \ref{pro1},
$B$ has the form 
\[B=\left[\begin{array}{c|c}B' & \\ \hline & \end{array}\right],\]
where $B'$ is a real matrix of order $2k \times 2s$ and rest
of the entries are zero matrices of appropriate order.
\par Consider the marginal Gaussian state, whose covariance
matrix is 
\begin{equation}\label{e:s}
\begin{bmatrix}A' & B' \\ B'^T & C' \end{bmatrix}.
\end{equation}
Since $A'$ and $C'$ do not have any principal sub-matrix of
the form $\frac{1}{2}I_2$, by Lemma \ref{le1}, the marginal
Gaussian state with covariance matrix given by (\ref{e:s}) is
completely extendable if and only if there is a real $2s
\times 2s$ matrix $\theta'$ such that 
\[C'+ \frac{\imath}{2} J_{2s} \ge \theta' \ge B'^T \left( A'
+ \frac{\imath}{2} J_{2k} \right)^{-1} B'.\]
Observe that 
\begin{align*}
B^T\left(A +\frac{\imath}{2}J_{2n} \right)^- B &=
\left[\begin{array}{c|c}B'^T & \\ \hline &
\end{array}\right] \left(
\left(A'+\frac{\imath}{2}J_{2k}\right)^{-1} 
\bigoplus \left( \bigoplus_{(m-k)\text{-copies}} \begin{bmatrix} \frac{1}{2}
& \frac{\imath}{2} \\ -\frac{\imath}{2} & \frac{1}{2}
\end{bmatrix} \right) \right)\left[\begin{array}{c|c}B' & \\
\hline & \end{array}\right] \\
&= \left[\begin{array}{c|cc}  B'^T \left( A' + \frac{\imath}{2} J_{2k}
\right)^{-1} B' & \bm{0}_{2s \times 2(n-s)}\\\hline
\bm{0}_{2(n-s) \times 2s}&\bm{0}_{2(n-s) \times 2(n-s)} \end{array}\right], 
\end{align*}
$\bm{0}$ with indices denoting zero matrices. 
Set $\theta = \theta' \bigoplus \bm{0}_{2(n-s) \times 2(n-s)}$. It is easy to see
that such a real matrix $\theta$ satisfies the conditions of inequality
(\ref{eq:xx}). Hence the theorem is proved. 
\end{proof}
\begin{theorem}\label{th:2}
Any separable Gaussian state in a bipartite system is
completely extendable.
\end{theorem}
\begin{proof}
Let $\rho$ be an $(m+n)$ mode Gaussian state with covariance
matrix $\begin{bmatrix} A & B^T \\ B & C \end{bmatrix}$
with $A$ and $C$ being the $m$ and $n$-mode marginal
covariance matrices.  By a theorem of Werner and
Wolf \cite{PhysRevLett.86.3658}, $\rho$ is
separable if and only if there exist $m$-mode and $n$-mode
Gaussian states with covariance matrices $X$ and $Y$
respectively such that
\[\begin{bmatrix} A & B^T \\ B & C \end{bmatrix} \ge
\begin{bmatrix} X & \\ & Y \end{bmatrix}.\]
Set $E=A-X$, $G=C-Y$, and $F=B$. Then the above inequality
can be expressed as
\[\begin{bmatrix} E & F^T \\F & G \end{bmatrix}\ge 0.\]
By the previous discussions and Theorem \ref{th:1}, we need
to construct a real, symmetric, $n\times n$ matrix $\varphi$
such that for every  $k$-extension, the matrix 
\[
\begin{bmatrix} 
E & F^T & F^T & F^T & \cdots & F^T\\
F & G & \varphi & \varphi & \cdots & \varphi\\
F & \varphi & G & \varphi & \cdots & \varphi\\
F & \varphi & \varphi & G & \cdots & \varphi\\
\vdots & \vdots & \vdots & \vdots & \ddots & \vdots\\
F & \varphi& \varphi& \varphi & \cdots & G
\end{bmatrix}\ge 0.\]
Calculations similar to those in Theorem \ref{th:1}, show
that this is possible if and only if 
\begin{equation}\label{th2:e1}
G \ge \varphi \ge F E^- F^T.
\end{equation}
We choose 
\begin{equation} \label{th2:e2}
\varphi = t G + (1-t) F E^- F^T, \quad t\in[0,1].
\end{equation}
Notice that for every $k=1,2,\cdots$, 
\[ \begin{bmatrix} X & && \\& Y && \\ &&\ddots & \\ &&&Y
\end{bmatrix} + \begin{bmatrix} 
E & F^T & F^T & F^T & \cdots & F^T\\
F & G & \varphi & \varphi & \cdots & \varphi\\
F & \varphi & G & \varphi & \cdots & \varphi\\
F & \varphi & \varphi & G & \cdots & \varphi\\
\vdots & \vdots & \vdots & \vdots & \ddots & \vdots\\
F & \varphi& \varphi& \varphi & \cdots & G
\end{bmatrix} = 
\begin{bmatrix} 
A & B^T & B^T & B^T & \cdots & B^T\\
B & C & \varphi & \varphi & \cdots & \varphi\\
B & \varphi & C & \varphi & \cdots & \varphi\\
B & \varphi & \varphi & C & \cdots & \varphi\\
\vdots & \vdots & \vdots & \vdots & \ddots & \vdots\\
B & \varphi& \varphi& \varphi & \cdots & C
\end{bmatrix},
\]
where the first term in the left hand side, $X
\oplus ( \oplus_k Y)$, is a Gaussian covariance matrix, and
the second one is a positive matrix thanks to the construction
above. Thus the right hand side is also a Gaussian covariance
matrix. Hence the theorem is proved with the extension
matrix $\varphi$ satisfying equation (\ref{th2:e2}).
\end{proof}

\begin{theorem}\label{th:3}
Any completely extendable Gaussian state is separable.
\end{theorem}
\begin{proof}
Let $S=\begin{bmatrix} A & B \\ B^T & C \end{bmatrix}$ be the
covariance matrix of a $(m+n)$ mode Gaussian state $\rho$, 
which is completely
extendable by a real symmetric positive matrix $\theta$
satisfying inequalities (\ref{eq:xx}) of Theorem \ref{th:1}. 
By the result of  Werner and Wolf
\cite{PhysRevLett.86.3658}, it is enough to find $m$ mode
and $n$ mode Gaussian states with covariance matrices $X$
and $Y$ respectively such that 
\begin{equation}\label{eq:y}
\begin{bmatrix} A & B \\ B^T & C \end{bmatrix} \ge
\begin{bmatrix} X & \\ & Y \end{bmatrix}.
\end{equation}
Since $S$ is extendable, $\theta \ge B^T \left( A +
\frac{\imath}{2} J_{2m} \right)^{-} B$, which is equivalent
to the matrix condition 
\[\begin{bmatrix} A + \frac{\imath}{2} J_{2m} & B \\ B^T &
\theta  \end{bmatrix} \ge 0.\]
Using Theorem \ref{cu:C}, this is equivalent to the condition that there is a
contraction $K$ with $\|K\|\le 1$ such that $B = \left( A +
\frac{\imath}{2} J_{2m}\right)^{\frac{1}{2}} K
\theta^{\frac{1}{2}}$.  Here
$\theta^{\frac{1}{2}} = \theta'^{\frac{1}{2}} \bigoplus
\bm{0}$,  where $\bm{0}$ being the zero matrix of appropriate
order, as in the proof of Theorem \ref{th:2}. In a similar way, we may define
$\theta^{-\frac{1}{2}} = \theta'^{-\frac{1}{2}} \oplus
\bm{0}$, which is a real positive matrix. Since $B$ and
$\theta$ are real matrices, so is $B\theta^{-\frac{1}{2}} = \left( A +
\frac{\imath}{2} J_{2m}\right)^{\frac{1}{2}} K
$. Hence, we can choose the contraction $K$ to be such that
the product in the right hand side is a real matrix.  Choose 
\begin{eqnarray*}
X &=& A - \left( A + \frac{\imath}{2}
J_{2m}\right)^{\frac{1}{2}} K K^\dag \left( A +
\frac{\imath}{2} J_{2m}\right)^{\frac{1}{2}}, \\
Y &=& C-\theta.
\end{eqnarray*}
By our choice of contraction $K$, $X$ chosen above is a real matrix. 
It follows from the left half of inequalities (\ref{eq:xx})
that $Y$ is a Gaussian covariance matrix. To see that $X$ is
so, observe that $\|K\|\le 0$, and  note that 
\begin{eqnarray*}
X+ \frac{\imath}{2} J_{2m} &=& \left( A + \frac{\imath}{2}
J_{2m}\right) -\left( A + \frac{\imath}{2}
J_{2m}\right)^{\frac{1}{2}} KK^\dag \left( A +
\frac{\imath}{2} J_{2m}\right)^{\frac{1}{2}} \\
&=&  \left( A + \frac{\imath}{2}
J_{2m}\right)^{\frac{1}{2}} (I - KK^\dag) \left( A +
\frac{\imath}{2} J_{2m}\right)^{\frac{1}{2}} \ge 0.
\end{eqnarray*}
To check that our choice of $X$ and $Y$ satisfies inequality
(\ref{eq:y}), we observe that -
\begin{eqnarray*}
\begin{bmatrix} A & B \\ B^T & C \end{bmatrix} -
\begin{bmatrix} X & \\ & Y \end{bmatrix} &=&
\begin{bmatrix}
\left( A + \frac{\imath}{2}
J_{2m}\right)^{\frac{1}{2}} K K^\dag \left( A +
\frac{\imath}{2} J_{2m}\right)^{\frac{1}{2}} & \left( A +
\frac{\imath}{2} J_{2m}\right)^{\frac{1}{2}} K
\theta^{\frac{1}{2}} \\
\theta^{\frac{1}{2}}K^\dag  \left( A +
\frac{\imath}{2} J_{2m}\right)^{\frac{1}{2}} & \theta
\end{bmatrix}\\
&=& \begin{bmatrix} \left( A + \frac{\imath}{2}
J_{2m}\right)^{\frac{1}{2}} & \\ & \theta^{\frac{1}{2}}
\end{bmatrix} 
\begin{bmatrix} K \\ I \end{bmatrix} \begin{bmatrix} K^\dag & I \end{bmatrix} 
\begin{bmatrix} \left( A + \frac{\imath}{2}
J_{2m}\right)^{\frac{1}{2}} & \\ & \theta^{\frac{1}{2}}
\end{bmatrix} \ge 0.
\end{eqnarray*}
Hence the theorem is proved.  
\end{proof}
We combine Theorem \ref{th:2} and \ref{th:3} to get a
necessary and sufficient condition for separability of
Gaussian states in the following theorem.
\begin{theorem}\label{th:4}
Any bipartite Gaussian state $\rho$ in $\Gamma(\mathbb{C}^m)
\otimes \Gamma(\mathbb{C}^n)$ is separable if and only if it
is completely extendable. 
\end{theorem}
\section{Complete extendability and separability in general
case}\label{sec:3}

\par Consider a separable Hilbert space $\mathfrak{h}$ and
denote $\mathcal{B} =\mathcal{B}(\mathfrak{h})$ the C*
algebra of all bounded operators on $\mathfrak{h}$. Let
$\mathcal{B}_n = \mathcal{B}(\mathfrak{h}^{\otimes n}) =
\mathcal{B}^{\otimes n}$ be the $n$-fold tensor product of
copies of $\mathcal{B}$. Then $\mathcal{B}_n$ can be
embedded as a C* algebra of $\mathcal{B}_{n+1}$ by the map
$X \mapsto X \otimes I,\, X\in \mathcal{B}_n, \, I$ being
the identity operator in $\mathcal{B}$. This enables the
construction of an inductive limit C* algebra
$\mathcal{B}^\infty$ such that there exists a C* embedding
$i_n: \mathcal{B}_n \hookrightarrow \mathcal{B}^\infty$ such
that $i_{n-1}  (X) = i_n (X \otimes I)$ for all $X \in
\mathcal{B}_{n-1},\, n=2,3,\cdots$. Let $\mathfrak{S}$
denote the set of all states in $\mathcal{B}^\infty$ equipped
with the weak* topology. Then $\mathfrak{S}$ is a compact
convex set. For any $\omega \in \mathfrak{S}$, define
\[\omega_n(X) = \omega(i_n(X)), \quad X \in \mathcal{B}_n.\]
Then $\omega_n$ is a state in $\mathcal{B}_n$ for all $n$
and 
\[\omega_{n-1}(X)= \omega_n(X \otimes I),\quad \forall X \in
\mathcal{B}_{n-1},\, n=2,3,\cdots.\]
in other words $\{\omega_n\}$ is a consistent family of
states in $\{\mathcal{B}_n\}$, $ n=2,3,\cdots$ with the
projective limit $\omega$. 

\par Conversely, let $\omega_n$ be a state in
$\mathcal{B}_n$ for each $n=1,2,3,\cdots$ such that
$\omega_n(X \otimes I) =\omega_{n-1}(X \otimes I),\, \forall
X \in \mathcal{B}_{n-1},\, n=2,3,\cdots$. Then there exists
a unique state $\omega$ in $\mathcal{B}^\infty$ such that 
\[\omega(i_n(X)) = \omega_n(X),\quad \forall X \in
\mathcal{B}_n,\, n=1,2,3,\cdots.\]

\begin{defin}
A state $\omega$ in $\mathcal{B}^\infty$ is said to be
\emph{locally normal} if each $\omega_n$ in $\mathcal{B}_n,
\, n=1,2,\cdots$ is determined by a density operator
$\rho_n,\, n=1,2,\cdots$, i.e., a positive operator $\rho_n$
of unit trace in $\mathfrak{h}^{\otimes n}$ satisfying 
\[\omega_n(X) = \tr \rho_nX, \quad X \in \mathcal{B}_n,\,
n=1,2,\cdots.\]
Then the relative trace of $\rho_n$ in
$\mathfrak{h}^{\otimes n}$  over the last copy of
$\mathfrak{h}$ is equal to $\rho_{n-1}$ for each
$n=2,3,\cdots.$
\end{defin}

\begin{defin}
A state in $\mathcal{B}^\infty$ is said to be
\emph{exchangeable}  if for any permutation $\pi$ of
$\{1,2,\cdots, n\}$ and operators $X_j \in \mathcal{B}, \,
i=1,2,\cdots,n$
\[\omega_n(X_{\pi(1)} \otimes X_{\pi(2)} \otimes \cdots
\otimes X_{\pi(n)}) = \omega_n (X_1 \otimes X_2 \otimes
\cdots \otimes X_n) = \omega( i_n(X_1 \otimes X_2 \otimes
\cdots \otimes X_n)).\]
\end{defin}
 
\par We shall now describe a version of quantum de Finetti
theorem due to Hudson and Moody \cite{MR0397421} (see also
St{\o}rmer \cite{MR0241992} for an abstract C* algebraic
version) which we shall make use of in our analysis of
complete extendability - separability problem. To this end
denote by $\mathcal{R}_\mathfrak{h}$ the set of all density
operators on $\mathfrak{h}$. Viewing
$\mathcal{R}_\mathfrak{h}$ as a subset of the dual of
$\mathcal{B}= \mathcal{B}_\mathfrak{h}$, equip it with the
relative topology inherited from the weak* topology. Let
$\mathcal{P}_\mathfrak{h}$ denote the set of all probability
measures on the Borel $\sigma$-algebra of
$\mathcal{R}_\mathfrak{h}$.

\begin{theorem} \label{th3.1}[Hudson and Moody]
A locally normal state $\omega$ on $\mathcal{B}^\infty$ is
exchangeable if and only if there exists a probability
measure $P_\omega$ in $\mathcal{P}_\mathfrak{h}$ such that 
\[\omega(i_n(X)) = \int_{\mathcal{R}_\mathfrak{h}} \tr
\rho^{\otimes n} X\, P_\omega(\dd \rho), \quad \forall X \in
\mathcal{B}_n,\, n=1,2,\cdots.\]
The correspondence $\omega \rightarrow P_\omega$ between the
set of locally normal and exchangeable states and the set
$\mathcal{P}_\mathfrak{h}$ of probability measures on
$\mathcal{R}_\mathfrak{h}$ is bijective.
\end{theorem}
 
\begin{rem}
Theorem \ref{th3.1} shows that exchanbeability property
automatically implies that every finite dimensional
projection of $\omega$, namely $\omega_n$, is separable. It
is natural to expect that complete extendeability would
force separability. 
\end{rem}

\begin{theorem}\label{th3.2}
Let $\mathfrak{h}_0,\, \mathfrak{h}$ be Hilbert spaces with
$\dim \mathfrak{h}_0 >2$ and 
$\rho$ be a density operator in $\mathfrak{h}_0 \otimes
\mathfrak{h}$. Let $\mathcal{B}_{n]} =
\mathcal{B}(\mathfrak{h}_0 \otimes \mathfrak{h}^{\otimes
n}), \, n=0,1,2,\cdots$. Suppose there exist density
operators $\rho_n$ in $\mathfrak{h}_0 \otimes
\mathfrak{h}^{\otimes n},\, n=1,2,\cdots$ satisfying the
following properties:
\begin{enumerate}
\item $\rho_1 =\rho$ and 
\[\tr \rho_n(X \otimes I) = \tr \rho_{n-1} X, \quad X \in
\mathcal{B}_{n]},\]
$I$ being the identity in $\mathfrak{h},\, n=1,2,\cdots$.
\item For any $X_0 \in \mathcal{B}(\mathfrak{h}_0), \, Y_j
\in \mathcal{B}(\mathfrak{h}),\, j=1,2,\cdots,n$ and any
permutation $\pi$ of $\{1,2,\cdots,n\}$
\[ \tr \rho_n \, X_0 \otimes Y_1 \otimes \cdots \otimes Y_n
= \tr \rho_n  \, X_0 \otimes Y_{\pi(1)} \otimes \cdots \otimes
Y_{\pi(n)}.\]
\end{enumerate}
Then $\rho$ is separable in $\mathfrak{h}_0 \otimes
\mathfrak{h}$. Furthermore $\rho_n$ is separable in
$\mathfrak{h}_0 \otimes \mathfrak{h}^{\otimes n},\, n=1,2,\cdots.$
\end{theorem}
\begin{proof}
We adopt the convention that for any density operator $\rho$
and any operator $X$ in a Hilbert space $\rho(X)=\tr \rho
X$.
Let $\mathcal{B}_n,\, i_n,\, \mathcal{B}^\infty$ be as
defined at the beginning of this section. Let $\rho_0$ be
the relative trace of $\rho$ over $\mathfrak{h}$ in
$\mathfrak{h}_0 \otimes \mathfrak{h}$. Choose and fix an
operator $0\le A \le I$ in $\mathfrak{h}_0$ such that
$\rho_0(A) = \tr \rho_0 A>0$. Then there exists a
well-defined state $\omega_A$ in the C* algebra
$\mathcal{B}^\infty$ satisfying 
\begin{equation}\label{eq:3.1}
\omega_A(i_n(Y)) =\frac{\rho_n(A \otimes Y)}{\rho_0(A)},
\quad Y \in \mathcal{B}_n,\, n=1,2,\cdots.
\end{equation}
Indeed, this follows from property (1) of $\{\rho_n\}$. Now
property (2) of $\rho_n$ implies that $\omega_A$ is
exchangeable   and locally normal. Thus by Theorem
\ref{th3.1} there exists a unique probability measure
$\mu_A$ on the Borel $\sigma$-algebra of
$\mathcal{R}_\mathfrak{h}$ such that 
\begin{equation}\label{eq:3.2}
\omega_A = \int_{\mathcal{R}_\mathfrak{h}} \sigma^\infty
\mu_A(\dd\sigma)
\end{equation}
where $\sigma^\infty$ denotes the unique state in
$\mathcal{B}^\infty$ satisfying 
\[\sigma^\infty(i_n(Y_1 \otimes \cdots \otimes Y_n)) =
\sigma(Y_1) \sigma(Y_2) \cdots \sigma(Y_n)\]
for all $Y_j \in \mathcal{B}(\mathfrak{h}),\, n=1,2,\cdots$.
\par In particular, the probability measure $\mu_I$ is
well-defined. When $\rho_0(A)=0$ we define $\mu_A$ to be
$\mu_I$. 
\par Equations (\ref{eq:3.1}) and (\ref{eq:3.2}) imply 
\begin{equation}\label{eq:3.3}
\rho_n(A \otimes Y) = \rho_0(A) \int_{\mathcal{R}_\mathfrak{h}}
\sigma^{\otimes n}(Y)\mu_A(\dd\sigma),\quad Y\in
\mathcal{B}_n,\, n=1,2,\cdots
\end{equation}
whenever $\rho_0(A) >0$. If $\rho_0(A)=0$ it follows from
the inequality $- \|Y\| A \otimes I \le A \otimes Y \le
\|Y\| A \otimes I$ that $\rho_n(A \otimes Y) =0$ whenever
$Y$ is self adjoint. Thus $\rho_n(A \otimes Y) =0$ for any
$Y \in \mathcal{B}_n$ whenever $\rho_0(A)=0$. Hence
(\ref{eq:3.3}) holds for all $0\le A \le I$.

\par Let $A\ge 0,\,B\ge 0,\,A+B \le I$ in
$\mathcal{B}(\mathfrak{h}_0)$. Writing down (\ref{eq:3.3})
for $A,\,B$ and $A+B$, adding the first two and comparing
with the third we get the relation 
\[\int \sigma^{\otimes n} (Y) \left(\rho_0(A) \mu_A +
\rho_0(B)\mu_B\right)(\dd\sigma) = \int \sigma^{\otimes n} (Y)
\rho_0(A+B) \mu_{A+B}(\dd\sigma),\quad Y\in \mathcal{B}_n,
\, n=1,2,\cdots.\]
If $\rho_0(A+B)>0$, dividing both sides by $\rho_0(A+B)$ and
using the uniqueness of the probability measure in Theorem \ref{th3.1}
we get 
\begin{equation}\label{eq:3.4}
\rho_0(A+B)\mu_{A+B} = \rho_0(A) \mu_A + \rho_0(B)\mu_B.
\end{equation}
If $\rho_0(A+B)=0$ then $\rho_0(A)=\rho_0(B)=0$ and hence
the same relation holds trivially. Choosing $0 \le A \le I,
\, B=I-A$ we have 
\[\mu_I= \rho_0(A)\mu_A + \rho_0(I-A)\mu_{I-A}.\]
This shows that as $A$ increases to $I$, the fact that
$\rho_0(I-A) \rightarrow 0$ implies that $\rho_0(A)
\mu_A(F)$ increases to $\rho_0(I)\mu_I(F)$. Let $\{
\bm{u}_j\}$ be an orthonormal basis for $\mathfrak{h}_0$.
For any unit vector $\bm{u} \in \mathfrak{h}_0$ and any
Borel set $F \subset \mathcal{R}_\mathfrak{h}$ define 
\begin{equation}\label{eq:3.5}
f(\bm{u}, F) = \rho_0(\ketbra{\bm{u}})
\mu_{\ket{\bm{u}}\bra{\bm{u}}} (F).
\end{equation}
Since $\sum_{j=1}^n \ketbra{\bm{u}_j}$ increases to the
identity operator as $n \rightarrow \infty$ it now follows
from (\ref{eq:3.4}) and the remark above
\[\sum_{j=1}^\infty f(\bm{u}_j, F) =\mu_I(F)\]
for any Borel set $F \subset \mathcal{R}(\mathfrak{h})$ and
any orthonormal basis. In other words, for each fixed $F$,
the map $\bm{u} \mapsto f(\bm{u},F)$ is a frame function in
the sense of Gleason \cite{MR0096113} on the unit sphere of
$\mathfrak{h}_0$. Hence by Gleason's theorem
\cite{MR0096113, krp4} there exists a positive trace class
operator $T(F)$ such that 
\[f(\bm{u},F) = \braket{\bm{u}|T(F)|\bm{u}},\quad \forall
\bm{u} \text{ with }\|\bm{u}\|=1\]
and any Borel set $F$ in $\mathcal{R}(\mathfrak{h})$. Thus
\begin{equation}\label{eq:3.6}
 \braket{\bm{u}|T(F)|\bm{u}} = \rho_0(\ketbra{\bm{u}})
\mu_{\ket{\bm{u}}\bra{\bm{u}}} (F).
\end{equation}
This together with (\ref{eq:3.3}) implies that $T(\cdot)$ is
a positive operator-valued measure satisfying 
\begin{eqnarray}
T(\mathcal{R}(\mathfrak{h})) &=& \rho_0,\nonumber\\
\tr T(F)A &=& \rho_0(A)\mu_A(F) \le \mu_I(F) \label{eq:3.7}
\end{eqnarray}
for any $0 \le A \le I$ in $\mathfrak{h}_0$ and any Borel
set $F$ in $\mathcal{R}(\mathfrak{h})$. Now choose and fix
an orthonormal basis $\{\bm{e}_j\}$ in $\mathfrak{h}_0$.
Then complex-valued measures
$\braket{\bm{e}_i|T(\cdot)|\bm{e}_j}, \, i,j=1,2,\cdots$ are
all absolutely continuous with respect to the measure
$\mu_I$ and hence there exist Radon-Nykodym derivatives
$f_{ij}$ in $\mathcal{R}(\mathfrak{h})$ satisfying the
relations 
\[\braket{\bm{e}_i|T(F)|\bm{e}_j} =\int_F f_{ij}(\sigma)
\mu_I(\dd\sigma).\]

\par The positivity of $T(F)$ for all Borel sets $F$ in
$\mathcal{R}(\mathfrak{h})$  implies that for any finite $n$
the matrix $((f_{ij}(\sigma))), \, i,j \in \{1,2,\cdots,n\}$
is positive semidefinite a.e. ($\mu_I$) for every $n$.
Furthermore,
\begin{eqnarray*}
\mu_I(F) &=& \tr T(F)\\
&=& \sum_{i=1}^\infty \braket{\bm{e}_i|T(F)|\bm{e}_i}\\
&=& \int_F \sum_i f_{ii}(\sigma) \mu_I(\dd\sigma)
\end{eqnarray*}
for all $F$. Thus 
\[ \sum_i f_{ii}(\sigma) =1,\quad \text{ a.e. } \mu_I.\]
Thus there exist density operators $\tau(\sigma), \,
\sigma\in \mathcal{R}_(\mathfrak{h})$ in $\mathfrak{h}_0$
such that 
\[ \braket{\bm{e}_i | \tau(\sigma) | \bm{e}_j}
=f_{ij}(\sigma)\]
a.e $\sigma(\mu_I)$.  Then 
\[T(F) =\int_F \tau(\sigma) \mu_I (\dd\sigma)\]
for every Borel set $F$ in $\mathcal{R}_{ \mathfrak{h}}$.
Now (\ref{eq:3.3}) and (\ref{eq:3.7}) imply that for any
$\ketbra{\bm{u}},\, \bm{u}$ a unit vector in $\mathfrak{h}_0$
\[\rho_n\left(\ketbra{\bm{u}} \otimes Y\right) =
\int_{\mathcal{R}_\mathfrak{h}} \sigma^{\otimes n} (Y)
\tau(\sigma) (\ketbra{\bm{u}}) \mu_I(\dd\sigma),\quad Y\in
\mathcal{B}_n,\, n=1,2,\cdots.\]
Thus 
\[\rho_n(A \otimes Y) = \int_{\mathcal{R}_\mathfrak{h}}
\tau(\sigma)(A) \sigma^{\otimes n} (Y) \mu_I(\dd\sigma)\]
for all $Y \in \mathcal{B}_n,\, A \in \mathcal{B}_0,\,
n=1,2,\cdots$. In other words each $\rho_n$ is separable in
the bipartite product $\mathfrak{h}_0 \otimes
[\mathfrak{h}]^{\otimes n}$. This completes the proof.
\end{proof}
\section{Conclusion}

\par Motivated by the notions of extendability  and complete
extendability of finite level states as described by Doherty
et al \cite{PhysRevA.69.022308} we introduce similar
definitions for Gaussian states. A necessary and sufficient
condition is obtained for the complete extendability of a
bipartite Gaussian state in terms of its covariance
matrix.Using only the properties of covariance matrices we
show the separability of any finite mode bipartite Gaussian
state is equivalent to its complete extendability.  The question of finding a
necessary and sufficient condition for $k$-extendability
remains open. By exploiting a version of the quantum de
Finetti theorem as in Hudson and Moody \cite{MR0397421}, and
Gleason's theorem \cite{MR0096113}, we prove the equivalence
of separability and complete extendibility of a bipartite
state whenever one of
the Hilbert spaces is of dimension greater than 2. Since the
C* algebra of all bounded operators of an infinite
dimensional separable Hilbert space is not separable, our
result is also an extension of Fannes, Lewis, and Verbeure \cite{MR948360}.

\bibliographystyle{acm}
\bibliography{biblio}

\end{document}